\documentclass[11pt]{article}

\usepackage[margin=1in]{geometry}
\usepackage{microtype}
\usepackage{xcolor}

\usepackage{enumitem}
\usepackage{graphicx}
\usepackage{fancybox}
\usepackage{comment}
\usepackage{xcolor}
\usepackage{nameref}
\usepackage{bbm}

\definecolor{ForestGreen}{rgb}{0.1333,0.5451,0.1333}
\definecolor{DarkRed}{rgb}{0.8,0,0}
\definecolor{Red}{rgb}{1,0,0}
\usepackage[linktocpage=true,
colorlinks,
linkcolor=ForestGreen,citecolor=ForestGreen,
bookmarks,bookmarksopen,bookmarksnumbered]
{hyperref}

\usepackage{amsmath}
\usepackage{amssymb}
\usepackage{amsthm}
\usepackage{mathtools}

\usepackage{subfig}

\usepackage{thm-restate}
\usepackage[ruled, noend, linesnumbered]{algorithm2e}

\usepackage{float}
\usepackage{cleveref}

\newcommand\vecone{\boldsymbol{1}}

\newcommand\R{\mathbb{R}}

\DeclareMathOperator*{\im}{im}

\newcommand{\eps}{\varepsilon}

\newcommand{\bDelta}{\boldsymbol{\Delta}}

\renewcommand{\hat}{\widehat}
\renewcommand{\tilde}{\widetilde}

\newcommand{\Otil}{\tilde{O}}

\renewcommand{\epsilon}{\ensuremath\varepsilon}

\newcommand\MM{\boldsymbol{\mathrm{{M}}}}

\newcommand\LL{\boldsymbol{\mathrm{{L}}}}

\renewcommand\aa{\boldsymbol{\mathrm{a}}}
\newcommand\bb{\boldsymbol{\mathrm{b}}}

\newcommand\ww{\boldsymbol{\mathrm{w}}}
\newcommand\xx{\boldsymbol{\mathrm{x}}}

\SetKw{Break}{break}
\SetKwComment{Comment}{$\triangleright$\ }{}
\SetCommentSty{mycommfont}

\newtheorem{theorem}{Theorem}[section]

\newtheorem{observation}[theorem]{Observation}

\newtheorem{claim}[theorem]{Claim}
\newtheorem{fact}[theorem]{Fact}

\newtheorem*{theorem*}{Theorem}

\newtheorem*{corollary*}{Corollary}
\newtheorem*{conjecture*}{Conjecture}
\newtheorem*{lemma*}{Lemma}
\newtheorem*{thm*}{Theorem}
\newtheorem*{prop*}{Proposition}
\newtheorem*{obs*}{Observation}
\newtheorem*{definition*}{Definition}

\newtheorem*{remark*}{Remark}
\newtheorem*{rec*}{Recommendation}

\usepackage[backend=biber, isbn=false, style=alphabetic, backref=true, doi=true, url=false, eprint=false, maxcitenames=10, mincitenames=3, maxbibnames=10, minbibnames=6, minalphanames=3, maxalphanames=5, defernumbers=true]{biblatex}

\usepackage{makecell}
\usepackage{booktabs}
\usepackage{array}
\usepackage[bottom]{footmisc}
\usepackage{tablefootnote}
 \newtheorem{invariant}[theorem]{Invariant}
\title{Partially-Dynamic All-Pairs Maxflow and Effective Resistance\\ via Stable Sparsifiers}
\author{
  Gramoz Goranci \\
  \small University of Vienna \\
  \small \texttt{gramoz.goranci@univie.ac.at}
  \and
  Rasmus Kyng\thanks{The research leading to these results has received funding from the starting grant “A New Paradigm for Flow and Cut Algorithms” (no. TMSGI2 218022) and grant no. 200021 204787 of the Swiss National Science Foundation.} \\
  \small ETH Zurich \\
  \small \texttt{kyng@inf.ethz.ch}
  \and
  Maximilian Probst Gutenberg\footnotemark[1] \\
  \small ETH Zurich \\
  \small \texttt{maximilian.probst@inf.ethz.ch}
  \and
  Yibin Zhao \\
  \small University of Toronto \\
  \small \texttt{ybzhao@cs.toronto.edu}
  \and
  Gernot Z\"ocklein\footnotemark[1] \\
  \small ETH Zurich \\
  \small \texttt{gernot.zoecklein@inf.ethz.ch}
}
\date{}

\begin{document}
\thispagestyle{empty}
\maketitle

\begin{abstract}
We give a randomized data structure for undirected weighted graphs that are partially dynamic, i.e., that undergo either only edge insertions or only edge deletions. The data structure maintains $(1\pm\eps)$-approximations to the maxflow value and effective resistance between any queried pair of vertices, with total update time $\widetilde{O}_{\eps}(n^2)$ and worst-case query time $\widetilde{O}_{\eps}(1)$.\footnote{We use $\tilde{O}(\cdot)$ notation to hide polylogarithmic factors in $n$ and $\tilde{O}_{\eps}(\cdot)$ notation to additionally suppress polynomial factors in $\eps^{-1}$. We let $m$ denote the maximum number of edges in the dynamic input graph at any time, and $n$ the number of vertices, and assume $m = \tilde{O}(n^2)$ polynomially-bounded weights throughout.} Thus, for dense graphs where $m = \Omega(n^2)$, our guarantees are near-optimal. Our algorithms succeed with high probability against an adaptive adversary.

Our result follows from a simple stability principle for partially dynamic graphs. We show how to partition an online sequence of $m$ updates into $\widetilde{O}(n/\eps)$ epochs such that every graph within an epoch is a $(1\pm O(\eps))$-spectral approximation of the graph at the beginning of the epoch. The epochs are determined by the cumulative leverage score of the updated edges: small leverage-score mass implies small spectral change, while the total leverage-score mass over a monotone update sequence is $\widetilde{O}(n)$. Consequently, a spectral sparsifier needs to be recomputed only once per epoch. Applying known static all-pairs maxflow and effective-resistance oracles to these sparsifiers then yields the result.
\end{abstract}

\pagebreak
\addtocounter{page}{-1}
\section{Introduction}

Maximum flow and effective resistance are two fundamental measures of
connectivity in a graph. The $st$-maxflow value measures how much flow the
network can route between two vertices $s$ and $t$, whereas the
$st$-effective resistance measures how well the network conducts between them.

\paragraph{Static algorithms and all-pairs oracles.}
In static graphs, both quantities are by now well understood. Exact $st$-maxflow can be computed in almost-linear $m^{1+o(1)}$ time \cite{chen2022maximum,van2023deterministic}, while a
$(1\pm\eps)$-approximation requires
$\widetilde{O}_{\eps}(m)$ time
\cite{peng2016approximate,fleischmann_et_al:LIPIcs.ICALP.2026.91}.
Similarly, a $(1\pm\eps)$-approximation to $st$-effective resistance can be
computed in near-linear time
\cite{spielman2014nearly,cohen2014solving,kyng2016approximate}.

Recent work has also studied \emph{all-pairs oracles}: data structures that,
after preprocessing, answer a query for any pair $s,t\in V$ in
$\widetilde{O}_{\eps}(1)$ time. Such an implicit representation is necessary
because explicitly reporting all pairwise values takes $\Omega(n^2)$ time.
For maxflow, exact oracles with $m^{1+o(1)}$ preprocessing time and $\widetilde{O}(1)$ query time are known
\cite{AbboudLPS23,AbboudKLPGSYY25}, as are $(1\pm\eps)$-approximate oracles with $\widetilde{O}_{\eps}(m)$ preprocessing time \cite{li2021approximate,li2023near}. For effective resistance, Spielman and Srivastava \cite{SpielmanS08} gave an elegant reduction to $\widetilde{O}{\eps}(1)$ Laplacian solves, yielding $\widetilde{O}{\eps}(m)$ preprocessing and
$\widetilde{O}_{\eps}(1)$ query time. Thus, in the static setting, both all-pairs problems admit essentially optimal preprocessing and query times.

\paragraph{Dynamic algorithms.} This raises the corresponding question for dynamic graphs:
\begin{center}    
\emph{Can all-pairs maxflow and effective resistance oracles be efficiently maintained\\ in graphs undergoing updates?}
\end{center}
However, the corresponding dynamic problem is conditionally much harder when exact answers are required. Assuming the OMv conjecture, Dahlgaard \cite{dahlgaard:LIPIcs.ICALP.2016.48} showed that exact $st$-maxflow in weighted undirected graphs cannot be maintained with $O(n^{1-\delta})$ amortized time per operation, for any constant $\delta>0$. This lower bound holds even for a fixed pair $s,t$ and even when the graph is only partially-dynamic, i.e. only undergoes edge insertions or deletions. One notable positive result is the algorithm of \cite{goranci2023efficient} for exact $st$-maxflow in unweighted graphs undergoing edge insertions, which obtains $\widetilde{O}(n^{5/2})$ total update time. Beyond this restricted setting, we are not aware of an exact dynamic algorithm that improves on solving the flow problem from scratch after every update.

It is therefore natural to allow for $(1\pm\eps)$-approximations to circumvent these strong lower bounds and to work with undirected graphs. A baseline solution follows from the implicit dynamic sparsifier in \cite{bernstein_et_al:LIPIcs.ICALP.2022.20} which yields $\tilde{O}_{\eps}(n)$ amortized update time. This is achieved by implicitly sampling a $(1\pm\eps)$-sparsifier consisting of $\tilde{O}_\eps(n)$ edges in $\tilde{O}_\eps(n)$ time and then running the above-mentioned static algorithms on the sparsifier. This yields algorithms to maintain oracles for both $(1\pm\eps)$-approximate all-pairs maxflow and all-pairs effective resistances.

In the most general, fully-dynamic setting, \cite{chen2020fast, goranci2021expander, van2024almost} culminated in a {deterministic}\footnote{Note that deterministic algorithms always work against an adaptive adversary.} $m^{o(1)}$-approximate all-pairs maxflow with $m^{1+o(1)}$ total update and $m^{o(1)}$ query time. \cite{chen2020fast} further gives an algorithm to maintain $(1\pm\eps)$-approximate all-pairs effective resistances with $\tilde{O}_{\eps}(mn^{2/3})$ total update time, breaking the above-stated $\Otil_{\epsilon}(mn)$ barrier, however, at the cost of a high query time of $\tilde{O}_{\eps}(n^{2/3})$. It also works only against an oblivious adversary. Thus, current fully dynamic algorithms improve on the baseline only by allowing for either a substantially weaker approximation or a substantially larger query time.

For the partially-dynamic setting, no improvement beyond these fully-dynamic algorithms is known for $(1\pm\eps)$-approximate all-pairs maxflow or effective resistances. Significant progress has instead been made for maxflow between
a \emph{fixed} pair $s,t$. These works obtain total update time
$\tilde{O}_{\eps}(m^{1+o(1)})$
\cite{van2024incremental,chen2024almostlinear,van2024almost} or
$\widetilde{O}_{\eps}(n^2)$
\cite{goranci2023efficient,goranci2025incremental,kravchenko2026partially}. Both approaches work against an adaptive adversary; the former is even deterministic.
Both heavily exploit the monotonicity of partially dynamic graphs: connectivity can only improve under insertions and only deteriorate under deletions. 

\paragraph{Our contribution.} In this paper, we break the update time baseline of $\widetilde{O}_{\eps}(n)$. We give a simple 
algorithm that maintains $(1\pm\eps)$-approximations to both all-pairs maxflow
values and all-pairs effective resistances in
$\widetilde{O}_{\eps}(n^2)$ total update time. 

\begin{theorem}\label{thm:main}
Let $G=(V,E,w)$ be an $n$-vertex undirected graph with polynomially-bounded positive edge weights undergoing either only edge insertions or only edge deletions, and let $\eps\in(0,1)$. There is a randomized data structure
that can be queried at any time, for any pair $s,t\in V$, to return
\begin{itemize}
\item a $(1\pm\eps)$-approximation to the $st$-maxflow value, and
\item a $(1\pm\eps)$-approximation to the $st$-effective resistance.
\end{itemize}
Each query takes worst-case time $\widetilde{O}(1/\eps^2)$, and the total
update time is $\widetilde{O}(n^2/\eps^7)$. The guarantees hold with high
probability against an adaptive adversary.
\end{theorem}

Our result matches the $\widetilde{O}{\eps}(n^2)$ total update time of the fixed-pair algorithms \cite{goranci2023efficient,goranci2025incremental,kravchenko2026partially} while supporting queries for arbitrary pairs. It is
near-optimal for dense monotone update sequences with
$m=\Theta(n^2)$ and improves on the $\widetilde{O}{\eps}(mn)$ baseline whenever $m=\omega(n)$. Compared with the previous fully dynamic all-pairs results, we obtain a $(1\pm\eps)$ rather than a subpolynomial approximation for maxflow, while retaining fast query time. For effective resistance, we additionally obtain guarantees against an adaptive adversary.

We note that the dependence on $\eps^{-1}$ in \label{thm:main} is larger than in the corresponding fixed-pair results. We believe that improving this dependence poses an interesting open problem.

\paragraph{Our approach.} In this article, we isolate a particularly simple consequence of monotonicity. We show that the online update sequence can be partitioned into $\tilde{O}(n/\eps)$ epochs. Within an epoch, graphs spectrally approximate each other closely, and so a sparsifier of the first graph in the epoch approximates all remaining graph version reasonably well. Thus, we can reduce the number of times that we sample a new sparsifier to $\tilde{O}(n/\eps)$ and thus only requires $\tilde{O}(n/\eps)$ many static all-pairs oracles rebuilds on such sparsifiers which then yields \label{thm:main} as corollary.

To partition epochs, we use a simple, elegant approach: we measure leverage scores of edges that are updated, and as long as edges touched by updates within the epochs have small combined leverage score mass, we know that graphs still spectrally approximate each other closely. On the other hand, since the total leverage score mass over all updates is bounded by $O(n \log n)$, we can show that we only need to start a new epoch $\tilde{O}_{\eps}(n)$ times.

\section{Preliminaries}

Given a graph $G = (V, E, w)$, we denote by $\LL_G = \sum_{e \in E} w_e \bb_e \bb_e^{\top}$ the Laplacian matrix of $G$. Here for an edge $e = (u, v)$ with some arbitrary but fixed orientation, we defined $\bb_e \in \R^V$ as $\bb_e = \vecone_v - \vecone_u$, where $\vecone_u$ is the vector with $1$ at entry $u$ and $0$ everywhere else. The effective resistance between two vertices $u$ and $v$ is defined as 
$$
R^{\text{Eff}}_{G}(u, v) := (\vecone_u - \vecone_v)^{\top} \LL_G^{\dag} (\vecone_u - \vecone_v),
$$
where $\LL_G^{\dag}$ is the Moore-Penrose Pseudo-Inverse of $\LL_G$. The effective resistance between $u$ and $v$ can be interpreted as the electrical energy it takes to route one unit of flow between $u$ and $v$ in $G$.

We say that a graph $H$ on the same vertex set as $G$ is a $(1\pm \eps)$-spectral sparsifier of $G$, also denoted by $H \approx_{\eps} G$, if
\[
    (1-\epsilon) \xx^\top \LL_G\xx
    \leq \xx^\top \LL_H\xx
    \leq (1+\epsilon)\xx^\top L_G\xx
    \qquad\text{for every }\xx\in\mathbb{R}^V
\]
where $\LL_G$ and $\LL_H$ are the Laplacian matrices associated with graphs $G$ and $H$, respectively. Using the Loewner order, the above condition can also succinctly be written as $(1-\epsilon) \LL_G \preceq \LL_H \preceq (1+\epsilon) \LL_G$.

The reason why we are interested in spectral sparsification is the following basic observation.
\begin{observation}
    Let $H$ be a $(1 \pm \epsilon)$ spectral sparsifier of $G$. Then for any pair of vertices $u, v \in V$, the effective resistance and maxflow value between $u$ and $v$ in $H$ is a $(1\pm \epsilon)$ approximation of the corresponding value in $G$.
\end{observation}
\begin{proof}
    For effective resistances, the above follows by definition. For the maxflow values, we use that if $H$ is a spectral sparsifier of $G$, then also all cuts are preserved up to $(1 \pm \epsilon)$ factors, which can be seen by plugging the indicator vectors of the cut into the Laplacian quadratic forms. The result on the maxflow value then follows from the maxflow min-cut theorem.
\end{proof}

\section{The Algorithm}

\subsection{Maintaining Flow Values via Adaptive Sparsifiers}

\paragraph{Statement of the Sparsifier Result.} The key idea behind our algorithm is that for any partially-dynamic graph $G$, we can maintain an adaptive spectral sparsifier $H$ of $G$ that undergoes changes during at most $\tilde{O}(n/\epsilon)$ time steps. Thus, it suffices to run a fast static algorithm to compute all-pairs effective resistances/ all-pairs maxflow values whenever $H$ changes. We point out that a stronger result, bounding even the number of changes to $H$ by $\tilde{O}(n/\epsilon^2)$, has been shown recently for incremental\footnote{We henceforth follow conventions to call graphs undergoing only edge insertions \emph{incremental}, and graphs that only undergo edge deletions \emph{decremental}.} graphs \cite{goranci2026onlinesparsificationalgorithmbook}. While using a similar algorithm, our weaker result has a significantly simpler proof. For decremental graphs, we are not aware of any previous algorithm that maintains a spectral sparsifier that changes at most $\Otil_{\epsilon}(n)$ many times and takes total update time $\Otil_{\epsilon}(n^2)$, even against an oblivious adversary. 

%this statement has only been known in the oblivious adversary setting. \gernot{has it even been known in the oblivious adversary setting? I dont think so?}\mprobst{Good point, I am also not really aware of this}

\begin{theorem}\label{thm:stableSparsifier}
Let $G=(V,E,w)$ be an undirected graph with polynomially bounded positive
edge weights undergoing either only edge insertions or only edge deletions,
and let $\epsilon\in(0,1)$. There is a randomized algorithm that maintains
a weighted graph $H$ such that, at any time, $G \approx_{\eps} H$.

Moreover, $H$ has $\tilde{O}( n/\epsilon^3)$ edges, changes at only $\widetilde{O}(n/\epsilon)$ update steps, and can be maintained in total update time $\tilde{O}(n^2 / \eps^{4})$. 

The guarantees hold with high probability against an adaptive adversary.
\end{theorem}

We remark that the statement that $H$ changes at only $\Otil(n / \epsilon)$ update steps is strictly weaker than $H$ changing by at most $\Otil(n / \epsilon)$ many edges over the entire course of the algorithm. That is because, for a timestep in which $H$ changes, it can change by arbitrarily many edges, i.e., all edges of $H$ might change at this time. This distinction does not matter for our application. We further note that, within an epoch where $H$ is unchanged, all graph versions of $G$ are $\eps$-spectral approximations, and thus $O(\eps)$-spectral approximations of each other, by transitivity. 

\paragraph{Reducing Flow Maintenance to Sparsifier Maintenance.} To obtain \Cref{thm:main} for effective resistances, it suffices to combine \Cref{thm:stableSparsifier} with the following statement from \cite{SpielmanS08} which cleverly exploits the Johnson-Lindenstrauss Lemma to compute all-pairs effective resistances via few Laplacian solves.

\begin{theorem}[see \cite{SpielmanS08}, Static All-Pairs Effective Resistances]\label{thm:ERstatic}
Given a graph $G=(V, E, w)$ with polynomially bounded conductances $w$, there is a randomized algorithm that constructs in $\tilde{O}(m \log(1/\epsilon) / \epsilon^2)$ time an oracle such that any query $s,t \in V$ can be answered in $\tilde{O}(1/\epsilon^2)$ time and returns a $(1\pm\epsilon)$-approximation to the $st$-effective resistance in $G$. The algorithm succeeds with high probability.
\end{theorem}

To obtain \Cref{thm:main} for the maxflow values of $G$, it suffices to point out that any $(1\pm\epsilon)$-approximate spectral sparsifier is a $(1\pm\epsilon)$-approximate cut sparsifier and therefore, by the famous maxflow-mincut theorem, maxflow values in $G$ and $H$ are preserved up to a $(1\pm\epsilon)$ factor. Finally, our algorithm simply invokes the following theorem from \cite{li2021approximate, li2023near} whenever $H$ changes on the sparsifier $H$.

\begin{theorem}[Static All-Pairs maxflow]
Given a graph $G=(V, E, w)$ with polynomially bounded capacities $w$, there is a randomized algorithm that constructs in $\tilde{O}(m/\epsilon^3)$ time\footnote{The dependency on $\epsilon$ is not explicitly calculated in \cite{li2023near}. Our calculations show a dependency of $\epsilon^{-6}$ which can be improved to an $\epsilon^{-3}$ dependency using the recent result from \cite{li2025simple}.} an oracle such that any query $s,t \in V$ can be answered in $\tilde{O}(1)$ time and returns a $(1\pm\epsilon)$-approximation to the $st$-maxflow value in $G$. The algorithm succeeds with high probability.
\end{theorem}

\subsection{Maintaining an Adaptive Sparsifier}

Finally, we describe our algorithm to maintain a sparsifier against an adaptive adversary that only changes $\tilde{O}(n/\eps)$ times over the course of the algorithm. Before we state the algorithm, we briefly review a result on implicit sampling from dynamic graphs.

\paragraph{Implicit Sampling.} We use the following result about implicitly sampling a spectral sparsifier in a fully-dynamic graphs appears in the arXiv version \cite{bernstein2020fully} of \cite{bernstein_et_al:LIPIcs.ICALP.2022.20}.

\begin{theorem}[see \cite{bernstein2020fully}, Theorem 10.5.]\label{thm:implicitSampler} Given an $m$-edge weighted graph $G$ and approximation parameter $\eps \in (0,1)$, there exists a data structure $\textsc{ImplicitSampler}(G, \eps)$ that supports two operations:
\begin{itemize}
    \item \textsc{Update}($e$): if $e$ is currently in $G$, it is deleted; otherwise it is inserted.
    \item \textsc{ReturnSparsifier}(): returns a $(1\pm \eps)$-spectral sparsifier $H$ of the current graph $G$ consisting of $\tilde{O}(n/\eps^3)$ edges.
\end{itemize}
The data structure requires $\tilde{O}(m)$ preprocessing time, $\tilde{O}(1)$ amortized update time and $\tilde{O}(n/\eps^3)$ query time. 

The guarantees hold with high probability against an adaptive adversary.
\end{theorem}

The above algorithm is rather straightforwardly obtained by observing that edges in $r$-regular $\phi$-expanders have effective resistance $\tilde{O}(1/(\phi^2 \cdot r))$. This insight is combined with the by-now-standard machinery of maintaining expander decompositions (see \cite{saranurak2019expander}) and a clever batching scheme. We refer the interested reader to \cite{bernstein2020fully}.

\paragraph{The Algorithm.} We are now ready to state our algorithm, which hinges on an extremely simple insight: if little leverage score mass changes, we can simply keep the old sparsifier; if a lot changes, we need to implicitly re-sample it; however, this can only happen a few times. For our algorithm, at time $t > 0$, we define the current leverage score $\tau^{(t)}$ of the edge $e^{(t)} = (u^{(t)}, v^{(t)})$ touched by the $t$-th update by 
\[\tau^{(t)} := \min\{1, w_{e^{(t)}} \cdot R^{\text{Eff}}_{G^{(t-1)}}(u^{(t)}, v^{(t)})\}
\]
In the incremental setting the quantity $\tau^{(t)}$ is also known as the online leverage score, whereas in the decremental setting, they are the reverse online leverage scores \cite{cohen2016online, braverman2020near}.

\begin{algorithm}[H]
\caption{$\textsc{MaintainSparsifier}(G, \eps)$}
\label{algo:main_algo}
$\mathcal{I} \gets \textsc{ImplicitSampler}(G^{(0)}, \eps / 10)$.\\
$H^{(0)} \gets \mathcal{I}.\textsc{ReturnSparsifier}()$; $t' \gets 0$.\\
\For{update $e^{(t)} = (u^{(t)}, v^{(t)})$ to $G^{(t-1)}$}{
    $\mathcal{I}.\textsc{Update}(e^{(t)})$.\\
    Compute a $(1\pm \eps / 2)$-approximation $\hat{\tau}^{(t)}$ of $\tau^{(t)}$.\\
    \If(\tcp*[h]{Start new epoch.}){$\sum_{t' < t'' \leq t} \hat{\tau}^{(t'')} > \epsilon / 20$}{
        $H^{(t)} \gets \mathcal{I}.\textsc{ReturnSparsifier}()$; $t' \gets t$.
    }\Else{
        $H^{(t)} \gets H^{(t-1)}$.
    }
}
\end{algorithm} 

\paragraph{Maintaining Leverage Score Estimates.} It remains to address how we compute the leverage score estimates $\hat{\tau}^{(t)}$. To this end, we maintain the following invariant (with high probability).

\begin{invariant}\label{inv:mainInvariant}
After the $t$-th for-loop iteration, we have $H^{(t)} \approx_{\eps / 4} G^{(t)}$.
\end{invariant}

Thus, when computing $\hat{\tau}^{(t)}$, for any $t > 0$, we already have $H^{(t-1)} \approx_{\eps / 4} G^{(t-1)}$ readibly available. We maintain the data structure from \Cref{thm:ERstatic} on $H$ which can therefore report in $\tilde{O}(1/\eps^2)$ time a $(1\pm \eps / 10)$-approximate effective resistance estimate for the endpoints of $e^{(t)}$ on $H^{(t-1)}$. %(here a subtle detail is that we make the decision whether we want to keep $H^{(t-1)}$ as our sparsifier based on $\hat{\tau}^{(t)}$ so we cannot use information from $H^{(t)}$ yet). 

%While for decremental graphs, this immediately yields the current leverage score, for incremental graphs, we need to  use the following fact, which immediately gives us a formula for computing $\hat{\tau}^{(t)}$.  It follows almost immediately from the Sherman--Morrison identity.

%\begin{fact}
%For any graph $G$ and edge $e = (u,v)$ we have
%\[
%    R^{\text{Eff}}_{G \cup e}(u,v) = \frac{R^{\text{Eff}}_{G}(u,v)}{1 + R^{\text{Eff}}_{G}(u,v)} 
%\]
%where we use the convention that $\infty/(1+\infty) = 1$.
%\end{fact}

This yields a $(1\pm \eps / 4) \cdot (1\pm \eps / 10)$-approximate estimate of $R^{\text{Eff}}_{G^{(t-1)}}(u^{(t)},v^{(t)})$, which can be converted in $O(1)$ time to to a $(1\pm \eps / 4)(1\pm \eps / 10) \leq (1 \pm \eps / 2)$-approximate estimate $\hat{\tau}^{(t)}$ of $\tau^{(t)}$, as desired.

\paragraph{Analysis.} We need two simple, but clever insights to show that the algorithm performs in the way intended; both have been observed in previous work. We start with the insight that small leverage score mass changes cannot affect the spectral approximation guarantee by much. This almost immediately yields the proof of \Cref{inv:mainInvariant}.

\begin{fact}[Small leverage mass implies small spectral change]
\label{lem:leverage-mass}
Let $\MM\succeq 0$, and $\bDelta=\sum_{i\in I}\aa_i \aa_i^\top$
where every $a_i\in\im(\MM)$. If $\sum_{i\in I}\aa_i^\top \MM^\dagger \aa_i\leq\eta$, then $\Delta\preceq\eta \MM$. Consequently,
\[
    (1-\eta)\MM\preceq \MM-\bDelta\preceq \MM.
\]
\end{fact}
\begin{proof}[Proof of \Cref{inv:mainInvariant}.]
At the end of the for-loop, we have $\sum_{t' < t'' \leq t} \hat{\tau}^{(t'')} \leq \eps / 20$ since otherwise the if-condition was satisfied, and $t'$ was reset to $t$, yielding an empty sum at the end of the loop. As we have $(1\pm \epsilon /2)$-approximations $\hat{\tau}^{(t)}$ of $\tau^{(t)}$, this implies that 
$$\sum_{t' < t'' \leq t} {\tau}^{(t'')} \leq \frac{1}{1-\epsilon/2}\sum_{t' < t'' \leq t} \hat{\tau}^{(t'')} \leq \frac{1}{1 - \epsilon / 2} \eps / 20 \leq (1+\epsilon) \epsilon / 20 \leq \epsilon / 10.$$ We further recall that $\tau^{(t'')} = \ww_{e^{(t'')}} \cdot \bb_{e^{(t'')}}^{\top}\LL_{G^{(t'')}}^{\dagger}\bb_{e^{(t'')}} = (\sqrt{\ww_{e^{(t'')}}} \cdot \bb_{e^{(t'')}})^{\top}\LL_{G^{(t'')}}^{\dagger}(\sqrt{\ww_{e^{(t'')}}} \cdot \bb_{e^{(t'')}})$ where $\bb_{e^{(t'')}}$ is the weighted incidence vector of the edge $e^{(t'')}$ and $\LL_{G^{(t'')}}^{\dagger}$ denotes the pseudo-inverse of the Laplacian matrix associated with $G^{(t'')}$. Thus, matching the format of the above fact. We will next show that $G^{(t')} \approx_{\epsilon / 10} G^{(t)}$. We proceed by a case distinction.

For incremental graphs, $G^{(t)}$ spectrally dominates each previous graph, i.e., $\LL_{G^{(t')}} \preceq \LL_{G^{(t'')}}$ for every $t' \leq t''$. Using that $A \preceq B$ implies $B^{\dagger} \preceq A^{\dagger}$, we thus have
\[
\sum_{t' < t'' \leq t} \ww_{e^{(t'')}} \cdot  \bb_{e^{(t'')}}\LL_{G^{(t)}}^{\dagger}\bb_{e^{(t'')}} \leq
\sum_{t' < t'' \leq t} \ww_{e^{(t'')}} \cdot \bb_{e^{(t'')}}\LL_{G^{(t'')}}^{\dagger}\bb_{e^{(t'')}} = \sum_{t' < t'' \leq t} {\tau}^{(t'')} \leq \epsilon / 10.
\]
Thus, by \Cref{lem:leverage-mass}, $\LL_{G^{(t')}} = \LL_{G^{(t)}} - \sum_{t' < t'' \leq t} \ww_{e^{(t'')}} \cdot  \bb_{e^{(t'')}}\bb_{e^{(t'')}} \succeq (1 - \epsilon / 10) \LL_{G^{(t)}}$. This allows us to conclude that $G^{(t')} \approx_{\epsilon /10} G^{(t)}$

For decremental graphs, we use a similar argument. Since now $\LL_{G^{(t')}} \succeq \LL_{G^{(t'')}}$ for every $t' \leq t''$, we have that
\[
\sum_{t' < t'' \leq t} \ww_{e^{(t'')}} \cdot  \bb_{e^{(t'')}}\LL_{G^{(t')}}^{\dagger}\bb_{e^{(t'')}} \leq
\sum_{t' < t'' \leq t}  \ww_{e^{(t'')}} \cdot  \bb_{e^{(t'')}}\LL_{G^{(t''-1)}}^{\dagger}\bb_{e^{(t'')}} = \sum_{t' < t'' \leq t} {\tau}^{(t'')} \leq \epsilon / 10.
\]
Thus, again by \Cref{lem:leverage-mass}, $\LL_{G^{{(t)}}} = \LL_{G^{(t')}} - \sum_{t' < t'' \leq t}  \ww_{e^{(t'')}} \cdot \bb_{e^{(t'')}}\bb_{e^{(t'')}} \succeq (1 - \epsilon / 10) \LL_{G^{(t')}}$. We can conclude that also in the decremental case, $G^{(t')} \approx_{\epsilon /10} G^{(t)}$.

Now, by the guarantees of \Cref{thm:implicitSampler}, we have that $H^{(t)} = H^{(t')} \approx_{\epsilon / 10} G^{(t')}$. As we have also just shown that $G^{(t')} \approx_{\epsilon / 10} G^{(t)}$, we can conclude that $H^{(t)}$ is  a $(1 \pm \epsilon / 10) \cdot (1 \pm \epsilon / 10) \leq (1 \pm \epsilon / 4)$ spectral approximation of $G^{(t)}$, as desired.

\end{proof}

Finally, let us establish stability. 

\begin{claim}[Stability]\label{clm:stability}
The number of times that $H$ is re-initialized is $O(n \log n/\eps)$. 
\end{claim}
\begin{proof}
By \cite{cohen2016online}, we have that $\sum_{t} \tau^{(t)} = O(n \log n)$ for $G$ being incremental, and since this fact also applies to the sequence invoked backwards, it immediately extends to the decremental setting as well, which was also observed in the streaming context in \cite{braverman2020near}. Since $\hat{\tau}^{(t)}$ are $(1\pm \eps / 2)$-approximate estimates, we also have $\sum_{t} \hat{\tau}^{(t)} = O(n \log n)$.

But we only rebuild at time $t$ if $\sum_{t' < t'' \leq t} \hat{\tau}^{(t'')} > \eps / 20$ and since these periods are disjoint, we can have at most $O(n \log n /\eps)$ rebuilds.
\end{proof}

Combined, \Cref{inv:mainInvariant}, \Cref{clm:stability} and \Cref{thm:implicitSampler} then yield \Cref{thm:stableSparsifier}.

\newpage
\printbibliography

@InProceedings{cohen2016online,
  author =	{Cohen, Michael B. and Musco, Cameron and Pachocki, Jakub},
  title =	{{Online Row Sampling}},
  booktitle =	{Approximation, Randomization, and Combinatorial Optimization. Algorithms and Techniques (APPROX/RANDOM 2016)},
  pages =	{7:1--7:18},
  series =	{Leibniz International Proceedings in Informatics (LIPIcs)},
  ISBN =	{978-3-95977-018-7},
  ISSN =	{1868-8969},
  year =	{2016},
  volume =	{60},
  editor =	{Jansen, Klaus and Mathieu, Claire and Rolim, Jos\'{e} D. P. and Umans, Chris},
  publisher =	{Schloss Dagstuhl -- Leibniz-Zentrum f{\"u}r Informatik},
  address =	{Dagstuhl, Germany},
  URL =		{https://drops.dagstuhl.de/entities/document/10.4230/LIPIcs.APPROX-RANDOM.2016.7},
  URN =		{urn:nbn:de:0030-drops-66304},
  doi =		{10.4230/LIPIcs.APPROX-RANDOM.2016.7}
}

@inproceedings{saranurak2019expander,
  title={Expander decomposition and pruning: Faster, stronger, and simpler},
  author={Saranurak, Thatchaphol and Wang, Di},
  booktitle={Proceedings of the Thirtieth Annual ACM-SIAM Symposium on Discrete Algorithms},
  pages={2616--2635},
  year={2019},
  organization={SIAM}
}

@InProceedings{bernstein_et_al:LIPIcs.ICALP.2022.20,
  author =	{Bernstein, Aaron and van den Brand, Jan and Probst Gutenberg, Maximilian and Nanongkai, Danupon and Saranurak, Thatchaphol and Sidford, Aaron and Sun, He},
  title =	{{Fully-Dynamic Graph Sparsifiers Against an Adaptive Adversary}},
  booktitle =	{49th International Colloquium on Automata, Languages, and Programming (ICALP 2022)},
  pages =	{20:1--20:20},
  series =	{Leibniz International Proceedings in Informatics (LIPIcs)},
  ISBN =	{978-3-95977-235-8},
  ISSN =	{1868-8969},
  year =	{2022},
  volume =	{229},
  editor =	{Boja\'{n}czyk, Miko{\l}aj and Merelli, Emanuela and Woodruff, David P.},
  publisher =	{Schloss Dagstuhl -- Leibniz-Zentrum f{\"u}r Informatik},
  address =	{Dagstuhl, Germany},
  URL =		{https://drops.dagstuhl.de/entities/document/10.4230/LIPIcs.ICALP.2022.20},
  URN =		{urn:nbn:de:0030-drops-163611},
  doi =		{10.4230/LIPIcs.ICALP.2022.20}
}

@article{bernstein2020fully,
  title={Fully-dynamic graph sparsifiers against an adaptive adversary},
  author={Bernstein, Aaron and Brand, Jan van den and Gutenberg, Maximilian Probst and Nanongkai, Danupon and Saranurak, Thatchaphol and Sidford, Aaron and Sun, He},
  journal={arXiv preprint arXiv:2004.08432},
  year={2020}
}

@misc{goranci2026onlinesparsificationalgorithmbook,
      title={An Online Sparsification Algorithm from the Book}, 
      author={Gramoz Goranci and Rasmus Kyng and Maximilian Probst Gutenberg and Yibin Zhao and Gernot Zöcklein},
      year={2026},
      eprint={2607.23098},
      archivePrefix={arXiv},
      primaryClass={cs.DS},
      url={https://arxiv.org/abs/2607.23098}, 
}

@INPROCEEDINGS{AbboudLPS23,
  author={Abboud, Amir and Li, Jason and Panigrahi, Debmalya and Saranurak, Thatchaphol},
  booktitle={2023 IEEE 64th Annual Symposium on Foundations of Computer Science (FOCS)}, 
  title={All-Pairs Max-Flow is no Harder than Single-Pair Max-Flow: Gomory-Hu Trees in Almost-Linear Time}, 
  year={2023},
  volume={},
  number={},
  pages={2204-2212},
  doi={10.1109/FOCS57990.2023.00137}}

@article{AbboudKLPGSYY25,
  title={Deterministic Almost-Linear-Time Gomory-Hu Trees},
  author={Abboud, Amir and Kyng, Rasmus and Li, Jason and Panigrahi, Debmalya and Gutenberg, Maximilian Probst and Saranurak, Thatchaphol and Yuan, Weixuan and Yuan, Wuwei},
  journal={arXiv preprint arXiv:2507.20354},
  year={2025}
}

@article{SpielmanS08,
    author = {Spielman, Daniel and Srivastava, Nikhil},
    title = {Graph Sparsification by Effective Resistances},
    journal = {SIAM Journal on Computing},
    volume = {40},
    number = {6},
    pages = {1913-1926},
    year = {2011},
    doi = {10.1137/080734029},
    url = {http://arxiv.org/abs/0803.0929}
}

@inproceedings{braverman2020near,
  title={Near optimal linear algebra in the online and sliding window models},
  author={Braverman, Vladimir and Drineas, Petros and Musco, Cameron and Musco, Christopher and Upadhyay, Jalaj and Woodruff, David P and Zhou, Samson},
  booktitle={2020 IEEE 61st Annual Symposium on Foundations of Computer Science (FOCS)},
  pages={517--528},
  year={2020},
  organization={IEEE}
}

@inproceedings{goranci2025incremental,
  title={Incremental Approximate Maximum Flow via Residual Graph Sparsification},
  author={Goranci, Gramoz and Henzinger, Monika and R{\"a}cke, Harald and Sricharan, AR},
  booktitle={52nd International Colloquium on Automata, Languages, and Programming (ICALP 2025)},
  pages={91--1},
  year={2025},
  organization={Schloss Dagstuhl--Leibniz-Zentrum f{\"u}r Informatik}
}

@inproceedings{van2024almost,
  title={Almost-linear time algorithms for decremental graphs: Min-cost flow and more via duality},
  author={Van Den Brand, Jan and Chen, Li and Kyng, Rasmus and Liu, Yang P and Meierhans, Simon and Gutenberg, Maximilian Probst and Sachdeva, Sushant},
  booktitle={2024 IEEE 65th Annual Symposium on Foundations of Computer Science (FOCS)},
  pages={2010--2032},
  year={2024},
  organization={IEEE}
}

@inproceedings{chen2024almostlinear,
  author       = {Chen, Li and Kyng, Rasmus and Liu, Yang P. and Meierhans, Simon and Probst Gutenberg, Maximilian},
  title        = {Almost-Linear Time Algorithms for Incremental Graphs: Cycle Detection, SCCs, s-t Shortest Path, and Minimum-Cost Flow},
  booktitle    = {Proceedings of the 56th Annual ACM Symposium on Theory of Computing},
  series       = {STOC '24},
  year         = {2024},
  publisher    = {Association for Computing Machinery},
  address      = {New York, NY, USA},
  doi          = {10.1145/3618260.3649745},
  url          = {https://doi.org/10.1145/3618260.3649745}
}

@misc{chen2022maximum,
  title         = {Maximum Flow and Minimum-Cost Flow in Almost-Linear Time},
  author        = {Li Chen and Rasmus Kyng and Yang P. Liu and Richard Peng and Maximilian Probst Gutenberg and Sushant Sachdeva},
  year          = {2022},
  eprint        = {2203.00671},
  archivePrefix = {arXiv},
  primaryClass  = {cs.DS},
  doi           = {10.48550/arXiv.2203.00671}
}

@inproceedings{peng2016approximate,
  title={Approximate undirected maximum flows in o (m polylog (n)) time},
  author={Peng, Richard},
  booktitle={Proceedings of the twenty-seventh annual ACM-SIAM symposium on Discrete algorithms},
  pages={1862--1867},
  year={2016},
  organization={SIAM}
}

@InProceedings{dahlgaard:LIPIcs.ICALP.2016.48,
  author =	{Dahlgaard, S{\o}ren},
  title =	{{On the Hardness of Partially Dynamic Graph Problems and Connections to Diameter}},
  booktitle =	{43rd International Colloquium on Automata, Languages, and Programming (ICALP 2016)},
  pages =	{48:1--48:14},
  series =	{Leibniz International Proceedings in Informatics (LIPIcs)},
  ISBN =	{978-3-95977-013-2},
  ISSN =	{1868-8969},
  year =	{2016},
  volume =	{55},
  editor =	{Chatzigiannakis, Ioannis and Mitzenmacher, Michael and Rabani, Yuval and Sangiorgi, Davide},
  publisher =	{Schloss Dagstuhl -- Leibniz-Zentrum f{\"u}r Informatik},
  address =	{Dagstuhl, Germany},
  URL =		{https://drops.dagstuhl.de/entities/document/10.4230/LIPIcs.ICALP.2016.48},
  URN =		{urn:nbn:de:0030-drops-63289},
  doi =		{10.4230/LIPIcs.ICALP.2016.48}
}

@inproceedings{goranci2023efficient,
  title={Efficient Data Structures for Incremental Exact and Approximate Maximum Flow},
  author={Goranci, Gramoz and Henzinger, Monika},
  booktitle={50th International Colloquium on Automata, Languages, and Programming, ICALP 2023},
  pages={69},
  year={2023}
}

@inproceedings{van2023deterministic,
  title={A deterministic almost-linear time algorithm for minimum-cost flow},
  author={Van Den Brand, Jan and Chen, Li and Kyng, Rasmus and Liu, Yang P and Peng, Richard and Gutenberg, Maximilian Probst and Sachdeva, Sushant and Sidford, Aaron},
  booktitle={2023 IEEE 64th Annual Symposium on Foundations of Computer Science (FOCS)},
  pages={503--514},
  year={2023},
  organization={IEEE}
}

@article{spielman2014nearly,
  title={Nearly Linear Time Algorithms for Preconditioning and Solving Symmetric, Diagonally Dominant Linear Systems},
  author={Spielman, Daniel A and Teng, Shang-Hua},
  journal={SIAM Journal on Matrix Analysis and Applications},
  volume={35},
  number={3},
  pages={835--885},
  year={2014},
  publisher={Society for Industrial and Applied Mathematics}
}

@inproceedings{li2021approximate,
  title={Approximate gomory--hu tree is faster than n--1 max-flows},
  author={Li, Jason and Panigrahi, Debmalya},
  booktitle={Proceedings of the 53rd Annual ACM SIGACT Symposium on Theory of Computing},
  pages={1738--1748},
  year={2021}
}

@inproceedings{li2025simple,
  title={A Simple and Fast Algorithm for Fair Cuts},
  author={Li, Jason and Li, Owen},
  booktitle={International Conference on Integer Programming and Combinatorial Optimization},
  pages={400--411},
  year={2025},
  organization={Springer}
}

@inproceedings{li2023near,
  title={Near-linear time approximations for cut problems via fair cuts},
  author={Li, Jason and Nanongkai, Danupon and Panigrahi, Debmalya and Saranurak, Thatchaphol},
  booktitle={Proceedings of the 2023 Annual ACM-SIAM Symposium on Discrete Algorithms (SODA)},
  pages={240--275},
  year={2023},
  organization={SIAM}
}

@inproceedings{cohen2014solving,
  title={Solving SDD linear systems in nearly m log1/2 n time},
  author={Cohen, Michael B and Kyng, Rasmus and Miller, Gary L and Pachocki, Jakub W and Peng, Richard and Rao, Anup B and Xu, Shen Chen},
  booktitle={Proceedings of the forty-sixth annual ACM symposium on Theory of computing},
  pages={343--352},
  year={2014}
}

@inproceedings{kyng2016approximate,
  title={Approximate gaussian elimination for laplacians-fast, sparse, and simple},
  author={Kyng, Rasmus and Sachdeva, Sushant},
  booktitle={2016 IEEE 57th Annual Symposium on Foundations of Computer Science (FOCS)},
  pages={573--582},
  year={2016},
  organization={IEEE}
}

@inproceedings{kravchenko2026partially,
  title={Partially-Dynamic Maximum Flow in Dense Graphs},
  author={Kravchenko, Egor and Probst Gutenberg, Maximilian},
  booktitle={53rd International Colloquium on Automata, Languages, and Programming (ICALP 2026)},
  pages={133--1},
  year={2026},
  organization={Schloss Dagstuhl--Leibniz-Zentrum f{\"u}r Informatik}
}

@inproceedings{van2024incremental,
  title={Incremental approximate maximum flow on undirected graphs in subpolynomial update time},
  author={van den Brand, Jan and Chen, Li and Kyng, Rasmus and Liu, Yang P and Peng, Richard and Gutenberg, Maximilian Probst and Sachdeva, Sushant and Sidford, Aaron},
  booktitle={Proceedings of the 2024 Annual ACM-SIAM Symposium on Discrete Algorithms (SODA)},
  pages={2980--2998},
  year={2024},
  organization={SIAM}
}

@inproceedings{chen2020fast,
  title={Fast dynamic cuts, distances and effective resistances via vertex sparsifiers},
  author={Chen, Li and Goranci, Gramoz and Henzinger, Monika and Peng, Richard and Saranurak, Thatchaphol},
  booktitle={2020 IEEE 61st Annual Symposium on Foundations of Computer Science (FOCS)},
  pages={1135--1146},
  year={2020},
  organization={IEEE}
}

@inproceedings{goranci2021expander,
  title={The expander hierarchy and its applications to dynamic graph algorithms},
  author={Goranci, Gramoz and R{\"a}cke, Harald and Saranurak, Thatchaphol and Tan, Zihan},
  booktitle={Proceedings of the 2021 ACM-SIAM Symposium on Discrete Algorithms (SODA)},
  pages={2212--2228},
  year={2021},
  organization={SIAM}
}

@InProceedings{fleischmann_et_al:LIPIcs.ICALP.2026.91,
  author =	{Fleischmann, Henry and Li, George Z. and Li, Jason},
  title =	{{Faster Weak Expander Decompositions and Approximate Max Flow}},
  booktitle =	{53rd International Colloquium on Automata, Languages, and Programming (ICALP 2026)},
  pages =	{91:1--91:20},
  series =	{Leibniz International Proceedings in Informatics (LIPIcs)},
  ISBN =	{978-3-95977-428-4},
  ISSN =	{1868-8969},
  year =	{2026},
  volume =	{374},
  editor =	{Bhattacharya, Sayan and Nanongkai, Danupon and Benedikt, Michael and Puppis, Gabriele},
  publisher =	{Schloss Dagstuhl -- Leibniz-Zentrum f{\"u}r Informatik},
  address =	{Dagstuhl, Germany},
  URL =		{https://drops.dagstuhl.de/entities/document/10.4230/LIPIcs.ICALP.2026.91},
  URN =		{urn:nbn:de:0030-drops-264800},
  doi =		{10.4230/LIPIcs.ICALP.2026.91}
}

\end{document}